\documentclass[iop,njp,groupedaddress,superscriptaddress,amsmath,amssymbol,notitlepage,stfloats,onecolumn]{revtex4-1}

\usepackage[latin1]{inputenc}
\usepackage{amsthm}
\usepackage{amssymb}
\usepackage{graphicx}
\usepackage{amsmath}
\usepackage{bbold}
\usepackage{bbm}
\usepackage{nonfloat}
\usepackage{color}
\usepackage[pdftex]{hyperref}
\usepackage{braket}
\usepackage{graphicx}
\usepackage{dsfont}
\usepackage{mathdots}

\usepackage{enumerate}

\DeclareGraphicsExtensions{.jpeg, .png, .pdf}

\newtheorem{thm}{Theorem}
\newtheorem*{thm*}{Theorem}

\newtheorem*{prop*}{Proposition}
\newtheorem{lemma}[thm]{Lemma}
\newtheorem*{lemma*}{Lemma}
\newtheorem{cor}[thm]{Corollary}
\newtheorem*{cor*}{Corollary}

\newtheorem*{cj*}{Conjecture}

\newtheorem*{Def*}{Definition}

\theoremstyle{definition}
\newtheorem{rem}{Remark}

\newcommand{\bq}{\begin{equation*}}
\newcommand{\be}{\begin{equation}}
\newcommand{\eq}{\end{equation*}}
\newcommand{\ee}{\end{equation}}
\newcommand{\bmu}{\begin{multline*}}
\newcommand{\emu}{\end{multline*}}
\newcommand{\ban}{\begin{align*}}
\newcommand{\bal}{\begin{align}}
\newcommand{\ean}{\end{align*}}
\newcommand{\eal}{\end{align}}

\newcommand{\id}{{\mathds{1}}}

\newcommand*{\coloneqq}{\mathrel{\vcenter{\baselineskip0.5ex \lineskiplimit0pt \hbox{\scriptsize.}\hbox{\scriptsize.}}} =}

\begin{document}

\title{Gaussian entanglement revisited}

\author{Ludovico Lami}
\affiliation{F\'{\i}sica Te\`{o}rica: Informaci\'{o} i Fen\`{o}mens Qu\`{a}ntics, Departament de F\'{i}sica, Universitat Aut\`{o}noma de Barcelona, ES-08193 Bellaterra (Barcelona), Spain}

\author{Alessio Serafini}
\affiliation{Department of Physics \& Astronomy, University College London,
Gower Street, London WC1E 6BT, United Kingdom}

\author{Gerardo Adesso}
\affiliation{Centre for the Mathematics and Theoretical Physics of Quantum Non-Equilibrium Systems, School of Mathematical Sciences, The University of Nottingham,
University Park, Nottingham NG7 2RD, United Kingdom}


\begin{abstract}
We present a novel approach to the separability problem for Gaussian quantum states of bosonic continuous variable systems. We derive a simplified necessary and sufficient separability criterion for arbitrary Gaussian states of $m$ vs $n$ modes, which relies on convex optimisation over marginal covariance matrices on one subsystem only. We further revisit the currently known results stating the equivalence between separability and positive partial transposition (PPT) for specific classes of Gaussian states. Using techniques based on matrix analysis, such as Schur complements and matrix means, we then provide a unified treatment and compact proofs of all these results. In particular, we recover the PPT-separability equivalence for: (i) Gaussian states of $1$ vs $n$ modes; and (ii) isotropic Gaussian states. In passing, we also retrieve (iii) the recently established equivalence between separability of a Gaussian state and and its complete Gaussian extendability.

Our techniques are then applied to progress beyond the state of the art. We prove that: (iv) Gaussian states that are invariant under partial transposition are necessarily separable; (v) the PPT criterion is necessary and sufficient for separability for Gaussian states of $m$ vs $n$ modes that are symmetric under the exchange of any two modes belonging to one of the parties; and (vi) Gaussian states which remain PPT under passive optical operations can not be entangled by them either. This is not a foregone conclusion per se (since Gaussian bound entangled states do exist) and settles a question that had been left unanswered in the existing literature on the subject.

This paper, enjoyable by both the quantum optics and the matrix analysis communities, overall delivers technical and conceptual advances which
are likely to be useful for further applications in continuous variable quantum information theory, beyond the separability problem.
\end{abstract}

\maketitle

\section{Introduction}

Gaussian states have played a privileged role in quantum optics and bosonic field theories,
essentially since the very early steps of such theories,
due to their ease of theoretical description and relevance to experimental practice.
Over the last twenty years, such a privilege has
carried over to quantum information science, where Gaussian states form the
core of the `continuous variable' toolbox~\cite{introeisert,biblioparis,adesso07,weedbrook12,adesso14,bucco}.
The analysis of quantum Gaussian states from the information theoretic standpoint brought up new subtle elements and much previously unknown insight into their structure. For, while Gaussian dynamics may essentially be dealt with
entirely at the phase space level (typically by normal mode decomposition, so that Gaussian dynamics are often trivialised as `quasi-free' in field theory), the analysis of quantum
information properties requires one to confront the Hilbert space description of the quantum states.
Hence, while Gaussian dynamics might well be exactly solvable with elementary tools,
the properties of Gaussian states related to the Hilbert space and tensor product structures
are far from being equally transparent.

The problem of {\em Gaussian separability}, that is, determining whether a bipartite Gaussian state is separable or entangled~\cite{adesso07}, exemplifies such a situation very well.
Necessary and sufficient conditions for Gaussian separability in the general $m$ vs $n$-mode case are available~\cite{Werner01,Giedke01}, yet they are recast in terms of convex optimisation problems whose solution (albeit numerically efficient) does not admit, in general, a closed analytical form. For non-separable states, a closely related question is whether their entanglement is distillable or bound~\cite{HorodeckiBound}. In the case of arbitrary multimode bipartite Gaussian states, while entanglement can never be distilled by Gaussian operations alone~\cite{nogo1,nogo2,nogo3}, it is known that entanglement distillability under general local operations and classical communications is equivalent to violation of the positivity of the partial transposition (PPT) criterion~\cite{Giedke01,GiedkeQIC}. In turn, the PPT criterion, which is as well efficiently computable at the level of covariance matrices for Gaussian states, and is in general only necessary for separability~\cite{Peres}, has been proven to be also sufficient in some important cases, notably when the bipartite Gaussian state under examination pertains to  a $1$ vs $n$-mode system~\cite{Simon00, Werner01}, when it is `bi-symmetric'~\cite{Serafini05}, i.e.~invariant under local permutations of any two modes within any of the two subsystems, and when it is `isotropic', i.e.~with a fully degenerate symplectic spectrum of its covariance matrix~\cite{holwer,botero03,giedkemode}. Outside of these special families, bound entangled Gaussian states can occur, as first shown in the $2$ vs $2$-mode case in~\cite{Werner01}.

In this paper we provide significant advances towards the characterisation of separability and entanglement distillability in Gaussian quantum states.
On one hand, we revisit the existing results, providing in particular a new compact proof for the equivalence between PPT and separability in $1$ vs $n$-mode Gaussian states, which encompasses the seminal $1$ vs $1$-mode case originally tackled by Simon~\cite{Simon00} and its extension settled by Werner and Wolf~\cite{Werner01}. Key to our proof is  the intensive use of Schur complements, which have enjoyed applications in various areas of (Gaussian) quantum information theory~\cite{Giedke01,giedkemode,nogo1,nogo2,nogo3,eisemi,gian,Simon16,Lami16}, and --- as further reinforced by this work --- may be appreciated as a mathematical cornerstone for continuous variable quantum technology.


On the other hand, we derive a number of novel results. In particular, a marginal extension of the techniques applied in the aforementioned proof allow us to prove  that Gaussian states invariant under partial transposition are necessarily separable, a result previously known only for the partial transposition of qubit subsystems
\cite{sep 2xN}.
We then show that the $1$ vs $n$-mode PPT-separability equivalence can be further extended to a class of arbitrary bipartite multimode Gaussian states that we call `mono-symmetric', i.e., invariant under local exchanges of any two modes on one of the two subsystems (see Fig.~\ref{mononucleosi}). This result, which (to the best of our knowledge) is observed and proven here for the first time, generalises the case of bi-symmetric states studied in~\cite{Serafini05}, providing as a byproduct a simplified proof for the latter as well.

As for isotropic Gaussian states, in the traditional approach the sufficiency of PPT for their separability follows from
a well known `mode-wise' decomposition of pure-state covariance matrices~\cite{holwer,botero03,giedkemode}, and from the fact
that the covariance matrix of an isotropic state is just a multiple of the covariance matrix of a pure Gaussian state.
Here, we derive the sufficiency of the PPT criterion for isotropic Gaussian states following a completely different and arguably more direct
approach. Main ingredients of this novel proof are advanced matrix analysis tools such as the operator geometric mean, already
found to be useful in the context of quantum optics ~\cite{Lami16}.

We also consider the well known class of Gaussian passive operations (i.e., the ones that preserve the average number of excitations of the
input state, such as beam splitters and phase shifters), which play a central role in quantum optics~\cite{introeisert,adesso14,bucco}, and we prove that a bipartite Gaussian state that always remains PPT under such a set of operations must also always stay separable.
This novel result complements the seminal study of~\cite{passive}, in that the latter only considered the possibility of turning a PPT state into a non-PPT one through passive operations --- essentially, the question of generating distillable entanglement --- which is not the same as the question of generating inseparability, because Gaussian PPT bound entangled states do exist~\cite{Werner01}. Here we settle the latter, more general and fundamental question.

All the previous results enable us to substantially extend the range of equivalence between Gaussian separability and PPT in contexts of strong practical relevance. Last but not the least --- in fact, first and foremost in the paper ---  we address the separability problem directly, and derive a simplified necessary and sufficient condition for Gaussian separability. For a bipartite state, this requires convex optimisation over marginal covariance matrices on one subsystem only, yielding a significant simplification over the existing criteria, which instead require optimisation on both parties~\cite{Werner01,Giedke01,eisemi}.

\begin{figure}[tb]
\centering \includegraphics[width=12cm]{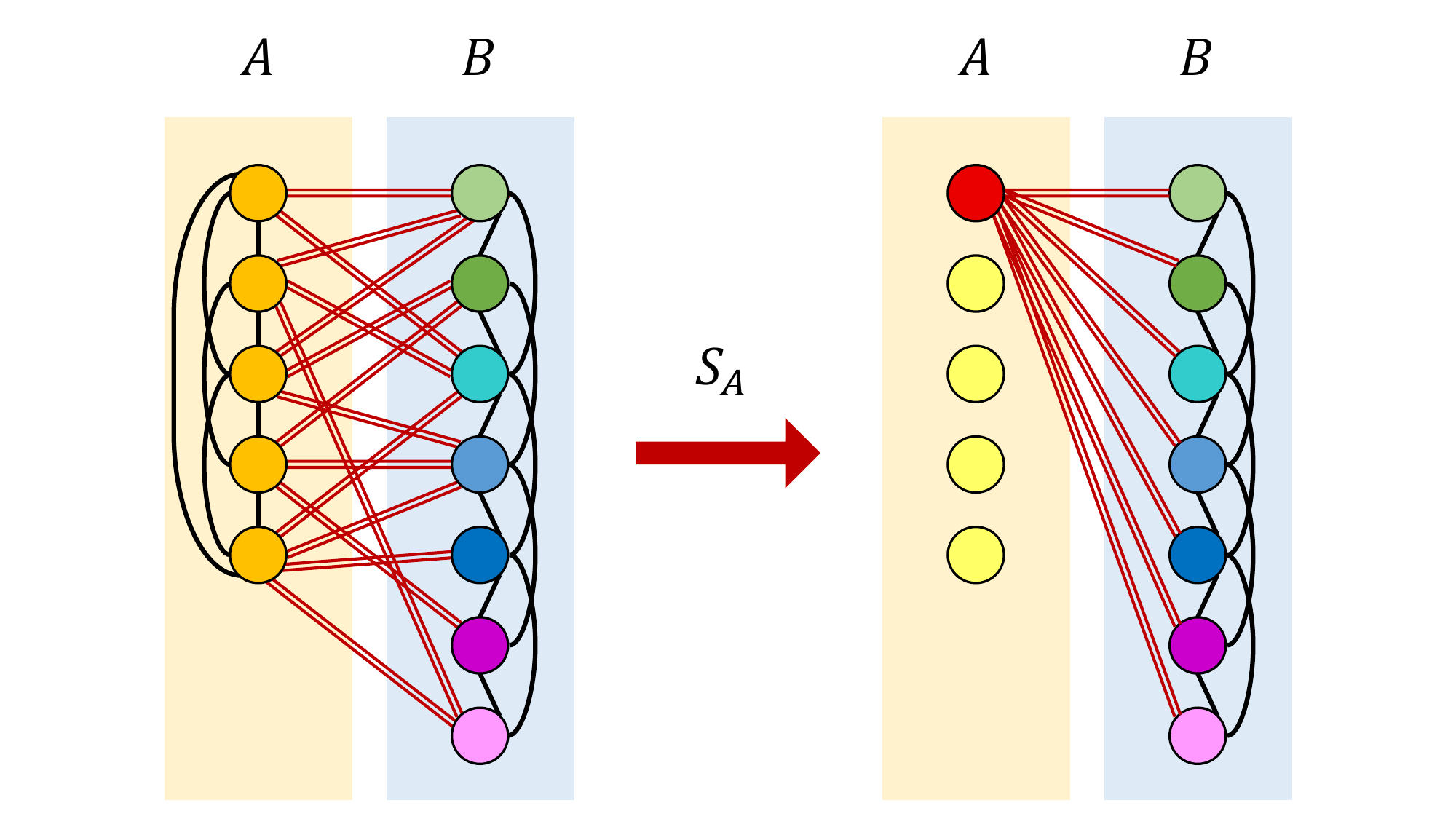}
\caption{Mono-symmetric Gaussian states of two parties $A$ (with $m$ modes) and $B$ (with $n$ modes) are invariant under exchange of any two modes within party $A$. By means of a suitable symplectic transformation on subsystem $A$, these states can be reduced to a $1$ vs $n$-mode Gaussian state and a collection of $m-1$ uncorrelated single-mode states on $A$'s side. Since PPT is equivalent to separability for $1$ vs $n$-mode Gaussian states, it follows that PPT is necessary and sufficient for separability of all $m$ vs $n$-mode mono-symmetric Gaussian states. In the schematics, entanglement between pairs of modes from the same party is depicted as a single solid (black) line, while entanglement across a mode from $A$ and a mode from $B$ is depicted as a double (dark red) line.
\label{mononucleosi}}
\end{figure}

This paper is organised as follows: in Sec.~\ref{methods}, the definition and basic properties of
Schur complements and matrix means that will be used in our derivations are recalled, and the Gaussian notation
is set; Sec.~\ref{secp} contains our main novel finding: a simplified necessary and sufficient condition for Gaussian separability in the general $m$ vs $n$-mode case; Sec.~\ref{ppt} contains a new proof of the sufficiency of the PPT criterion for separability of
$1$ vs $n$-mode Gaussian states achieved in a few, swift Schur complements'
manipulation;
Sec.~\ref{inva} shows that invariance under partial transposition implies separability;
in Sec.~\ref{symm} we prove that mono-symmetric states can be reduced, under local
unitary operations, to $1$ vs $n$-mode states, which implies that the PPT condition is sufficient for them too (see Fig.~\ref{mononucleosi});
in Sec.~\ref{mode-wise} we provide the reader with a new proof of the sufficiency of PPT for separability of isotropic Gaussian states, not relying on their mode-wise decomposition; Sec.~\ref{pass} contains our novel analysis of
entanglement generation under passive operations.
Sec.~\ref{outro} concludes the paper with a brief summary and some future perspectives related to this work.


\section{The toolbox: Schur complements, matrix means and Gaussian states} \label{methods}


One of the messages of the present paper is to lend further support to the fact that methods based on Schur complements and matrix means can be successfully exmployed to derive fundamental results in continuous variable quantum information, following a streak of applications to various contexts including separability, distillability, steerability, entanglement monogamy, characterisation of Gaussian maps, and related problems~\cite{Giedke01,giedkemode,nogo1,nogo2,nogo3,eisemi,gian,Simon16,Lami16}.
As a divertissement to set the stage, let us present a compact, essential compendium of such methods.
\vspace{3ex}
\subsection{Schur complements}\label{secSC}
Given a square matrix $M$ partitioned into blocks as
\begin{equation}
M\, =\, \begin{pmatrix} A & X \\ Y & B \end{pmatrix}\, , \label{block}
\end{equation}
the {\it Schur complement} of its (square, invertible) principal submatrix $A$, denoted by $M/A$, is defined as
\begin{equation}
M/A \coloneqq B - YA^{-1} X\, .
\label{schur}
\end{equation}
A useful reference on Schur complements is the monograph~\cite{ZHANG05}. Here we limit ourselves to stress some of the properties we will make use of in the present paper. As it turns out, Schur complements are the answer~\cite{42} to a number of questions that arise pretty naturally in matrix analysis. Many of these applications stem from the fact that the positivity conditions of $2\times 2$ hermitian block matrices can be easily written in terms of Schur complements.

\begin{lemma} \label{pos cond}
Consider a hermitian matrix
\begin{equation} H = \begin{pmatrix}  A & X \\ X^\dag & B \end{pmatrix} . \label{H part} \end{equation}
Then $H$ is strictly positive definite ($H>0$) if and only if $A > 0$ and $H/A=B-X^\dag A^{-1}X > 0$. Then, by taking suitable limits, $H$ is semidefinite positive ($H\geq 0$) if and only if $A\geq 0$ and $B-X^{\dag}(A+\varepsilon\mathds{1})^{-1} X\geq 0$ for all $\varepsilon>0$.
\end{lemma}

A consequence of this result that will be relevant to us is the following.

\begin{cor}
Let $H$ be a hermitian matrix partitioned as in~\eqref{H part}. Then, if $A>0$,
\begin{equation}  H/A\,=\, \sup \Big\{ \tilde{B}=\tilde{B}^\dag:\ H > 0\oplus \tilde{B} \Big\}\, . \label{variational} \end{equation}
Here we mean that the matrix set on the right-hand side has a supremum (i.e.~a minimum upper bound) with respect to the L\"owner partial order ($X\geq Y$ if and only if $X-Y$ is positive semidefinite), and that this supremum is given by the Schur complement on the left-hand side.
\end{cor}

We note in passing that from the above variational representation it follows immediately that $H/A$ is monotone and concave in $H>0$.

\subsection{Matrix means}\label{secMM}
{
Somehow related to Schur complements are the so-called matrix means. As one might expect from their name, these are functions taking two positive matrices as inputs and yielding another positive matrix as output. For an excellent introduction to this topic, we refer the reader to~\cite[Chapter 4]{BHATIA}. Here, we review only some basic facts that we will find useful throughout the paper. Given two strictly positive matrices $A,B>0$, the simplest mean one can define is the {\it arithmetic mean} $(A+B)/{2}$, whose generalisation from scalars to matrices does not present difficulties. Another easily defined object is the {\it harmonic mean}~\cite{parallel sum, ando79}, denoted by $A!B$ and given by
\begin{equation}
A!B\, \coloneqq\, \left(\frac{A^{-1}+B^{-1}}{2}\right)^{-1}\, . \label{harmonic}
\end{equation}
Incidentally, the harmonic mean can be also defined as a Schur complement, with the help of the identity $A!B=A-A(A+B)^{-1}A=\begin{pmatrix} A & A \\ A & A+B \end{pmatrix} \Big/ (A+B)$, which immediately implies that $A!B$ is monotone and jointly concave in $A$ and $B$, i.e.~concave in the pair $(A,B)$.

The least trivially defined among the elementary means is undoubtedly the {\it geometric mean} $A\# B$ between strictly positive matrices $A,B>0$~\cite{geometric mean, ando79}, which can be constructed as
\begin{equation}
A\# B\, \coloneqq\, \max\{X=X^{\dag}:\ A\geq XB^{-1} X\}\, , \label{geometric}
\end{equation}
where the above maximisation is with respect to the L\"owner partial order, and the fact that the particular set of matrices we chose admits an absolute maximum is already nontrivial. With a bit of work one can show that $A\# B$ is explicitly given by
\begin{equation}
A\# B\, =\, A^{1/2} \left( A^{-1/2} B A^{-1/2} \right)^{1/2} A^{1/2}\, . \label{geom expl}
\end{equation}
Having multiple expressions for a single matrix mean is always useful, as some properties that are not easy to prove within one formulation may become apparent when a different approach is taken. For instance, the fact that $A\# B$ is covariant under congruences, i.e.~$(MAM^{\dag})\#(MBM^{\dag})=M(A\#B)M^{\dag}$ for all invertible $M$, is far from transparent if one looks at~\eqref{geom expl}, while it becomes almost obvious when~\eqref{geometric} is used. On the contrary, the fact that $A\#B=(AB)^{1/2}$ when $[A,B]=0$ is not easily seen from~\eqref{geometric}, but it is readily verified employing~\eqref{geom expl}.

As it happens with scalars, the inequality
\begin{equation}
A!B\, \leq\, A\#B\, \leq\, \frac{A+B}{2} \label{hga}
\end{equation}
holds true for all $A,B>0$. In view of the above inequality, it could be natural to wonder, how the geometric mean between the leftmost and rightmost sides of~\eqref{hga} compares to $A\#B$. That this could be a fruitful thought is readily seen by asking the same question for real numbers. In fact, when $0<a,b\in\mathds{R}$ it is elementary to verify that $\sqrt{ab}=\sqrt{\frac{a+b}{2}\, \left( \frac{1/a+1/b}{2}\right)^{-1}}$. Our first result is a little lemma extending this to the non-commutative case. We were not able to find a proof in the literature, so we provide one.

\begin{lemma} \label{lemma ha=g}
For $A,B>0$ strictly positive matrices, the identity
\begin{equation}
A\#B\, =\, \left(\frac{A+B}{2}\right)\#\left( A!B \right)
\label{ha=g}
\end{equation}
holds true.
\end{lemma}

\begin{proof}
We start by defining $\tilde{A}\coloneqq \left(A+B\right)^{-1/2}A\left(A+B\right)^{-1/2}$, $\tilde{B}\coloneqq\left(A+B\right)^{-1/2}B\left(A+B\right)^{-1/2}$. It is easy to see that $[\tilde{A},\tilde{B}]=0$, for instance because $\tilde{A}+ \tilde{B}=\left(A+B\right)^{-1/2}(A+B)\left(A+B\right)^{-1/2}=\mathds{1}$. Therefore, the identity $\tilde{A}\#\tilde{B}=(\tilde{A}\tilde{B})^{1/2}$ holds. Now, on the one hand the congruence covariance of the geometric mean implies that
\begin{equation*}
\tilde{A}\#\tilde{B}\, =\, \left( \left(A+B\right)^{-1/2}A\left(A+B\right)^{-1/2} \right) \# \left( \left(A+B\right)^{-1/2}B\left(A+B\right)^{-1/2} \right)\, =\, \left(A+B\right)^{-1/2}(A\# B)\left(A+B\right)^{-1/2}\, .
\end{equation*}
On the other hand,
\begin{align*}
\tilde{A}\tilde{B}\, &=\, \left(A+B\right)^{-1/2}A\left(A+B\right)^{-1}B\left(A+B\right)^{-1/2}\, =\\
&=\, \left(A+B\right)^{-1/2}\left(B^{-1}(A+B)A^{-1}\right)^{-1}\left(A+B\right)^{-1/2} \, =\, \frac12\, \left(A+B\right)^{-1/2}(A!B)\left(A+B\right)^{-1/2}\, .
\end{align*}
Putting all together, we see that
\begin{equation*}
\left(A+B\right)^{-1/2}(A\# B)\left(A+B\right)^{-1/2}\, =\, \tilde{A}\#\tilde{B}\, =\, (\tilde{A}\tilde{B})^{1/2}\, =\, \frac{1}{\sqrt{2}}\, \left(\left(A+B\right)^{-1/2}(A!B)\left(A+B\right)^{-1/2}\right)^{1/2}\, .
\end{equation*}
Conjugating by $(A+B)^{1/2}$, we obtain
\begin{equation}
A\# B\, =\, \frac{1}{\sqrt{2}}\, \left(A+B\right)^{1/2} \left(\left(A+B\right)^{-1/2}(A!B)\left(A+B\right)^{-1/2}\right)^{1/2} \left(A+B\right)^{1/2}\, =\, \left(\frac{A+B}{2}\right)\#\left( A!B \right) ,
\end{equation}
where the last step is an application of~\eqref{geom expl}.
\end{proof}
}

\subsection{Gaussian states}\label{secGS}
In the remainder of this Section, we provide a brief introduction to the main concepts of the Gaussian formalism.
Quantum continuous variables just describe quantum mechanics applied to an infinite-dimensional Hilbert space equipped with position and momentum operators $x_j,p_k$ ($j,k=1,\ldots, n$) satisfying the so-called canonical commutation relations $[x_j,p_k]=i\delta_{jk}$ (in natural units, $\hbar=1$). Such a Hilbert space describes, for instance, a collection of $n$ quantum harmonic oscillators, or modes of the electromagnetic radiation field. The operators $x_j,p_k$ are often grouped together to form a single vector of $2n$ operators $r \coloneqq  (x_1,p_1,\ldots,x_n,p_n)$. The canonical commutation relations then take the form
\begin{equation}
[r,r^T]\, =\, i\ \Omega\, \coloneqq\, i \  \omega^{\oplus n}\,, \quad \omega \coloneqq   \begin{pmatrix} 0 & {1} \\ - {1} & 0 \end{pmatrix}.
\label{CCR}
\end{equation}
An important object one can form is the displacement operator. For any $z\in\mathds{R}^{2n}$, we define
\begin{equation}
D_z\, \coloneqq\, e^{iz^T\Omega r}\, .
\label{displacement}
\end{equation}

It turns out that Nature has a special preference for quadratic Hamiltonians. A prominent example is the free-field Hamiltonian $\mathcal{H}_0=\frac{1}{2} r^T r$. Not surprisingly, thermal states of quadratic Hamiltonians are extremely easily produced in the lab, in fact so easily that they deserve a special name: Gaussian states~\cite{introeisert,biblioparis,adesso07,weedbrook12,adesso14,bucco}. As the name suggests, they can be fully described by a real displacement vector $w\in\mathds{R}^{2n}$ and a real, $2n\times 2n$ {\it quantum covariance matrix} (QCM) $V$, defined respectively as $w=\langle r \rangle$ and $V = \langle \{r-w,r^T-w^T\}\rangle$. By quantum covariance matrix we mean a real, symmetric, strictly positive matrix $V>0$ that moreover satisfies the Heisenberg uncertainty relation~\cite{simon94}
\begin{equation}
V+i\Omega \geq 0\, . \label{Heisenberg}
\end{equation}
Note that (\ref{Heisenberg}) can equivalently be written as $V-i\Omega \geq 0$ upon applying transposition (as $V^T=V$, $\Omega^T=-\Omega$).

The Gaussian state $\rho^G(V,w)$ with QCM $V$ and displacement vector $w$ admits the representation
\begin{equation}
\rho^G(V,w)\, =\, \int \frac{d^{2n} u}{(2\pi)^n} \ e^{-\frac{1}{4} u^T V u - iw^T r} D_{\Omega r}\,,
\end{equation}
which justifies the alternative definition of Gaussian states as the continuous variable states associated with a Gaussian characteristic function.

Clearly, linear transformations $r \rightarrow S r$ that preserve the commutation relations~\eqref{CCR} play a special role within this framework. Any such transformation is described by a {\it symplectic} matrix, i.e.~a matrix $S$ with the property that $S\Omega S^T=\Omega$. Symplectic matrices form a non-compact, connected Lie group that is additionally closed under transposition, and is typically denoted by $\mathrm{Sp}(2n,\mathds{R})$~\cite{pramana}. The importance of these operations arises from the fact that for any symplectic $S$ there is a unitary evolution $U_S$ on the Hilbert space such that $U_S^\dag r U_S = Sr$. Most importantly, such a unitary is the product of a finite number of factors $e^{i\mathcal{H}_Q}$, where $\mathcal{H}_Q$ is a quadratic Hamiltonian, and as such it can be easily implemented in laboratory. Under conjugation by $U_{S}$, Gaussian states transform as
\begin{equation}
U_S^\dag\, \rho^G(V,w)\, U_S\, =\, \rho^G\left(SVS^T,\, Sw \right)\, .
\label{Gauss transform U_S}
\end{equation}

It turns out that all Gaussian states can be brought into a remarkably simple normal form via unitary transformations induced by quadratic Hamiltonians. In fact, a theorem by Williamson~\cite{willy,willysim} implies that for all strictly positive matrices $V>0$ there is a symplectic transformation $S$ and a diagonal matrix $N>0$ such that 
\begin{equation}
S^{-1} VS^{-T}\, =\, N \coloneqq \text{diag} (\nu_1,\nu_1,\ldots,\nu_n,\nu_n) \, .
\label{Williamson}
\end{equation}
The diagonal elements $\nu_i>0$, each taken with multiplicity one, are called symplectic eigenvalues of $V$, and are uniquely determined by $V$ (up to their order, which can be assumed decreasing by convention with no loss of generality). Accordingly, we will refer to $\vec{\nu}=(\nu_1,\nu_2,\ldots,\nu_n)$ as the {\em symplectic spectrum} of $V$. Notably, Heisenberg uncertainty relation~\eqref{Heisenberg} can be conveniently restated as $N\geq \mathds{1}$, or equivalently $\nu_{i}\geq 1$ for all $i=1,\ldots, n$. { A Gaussian state $\rho^G(V,w)$ can be shown to be pure if and only if all of its symplectic eigenvalues are equal to $1$, which corresponds to the matrix equality $V\Omega V\Omega=-\id$. Correspondingly, a QCM $V$ satisfying $\nu_{j}=1$ for all $j=1,\ldots,n$ (or equivalently $\det V = 1$) will be called a {\it pure} QCM. Note that pure QCMs $V$ are themselves symplectic matrices, $V=(S^T S)^{-1} \in \mathrm{Sp}(2n,\mathds{R})$, and are the extremal elements in the convex set of QCMs.}

Finally, note that displacement vector $w$ is often irrelevant since it can be made to vanish by local unitaries, resulting from the action of the displacement operator of (\ref{displacement}) on each individual mode. Since all the physically relevant informational properties such as purity and entanglement are invariant under local unitaries, all the results we are going to present will not depend on the first moments. Therefore, in what follows, we will completely specify any Gaussian state under our investigation as $\rho^G(V)$ in terms of its QCM $V$ alone.


\section{Simplified separability criterion for  M vs N-mode Gaussian states}\label{secp}

The QCM $V_{AB}$ of a Gaussian state $\rho^G_{AB}$ pertaining to a $(m+n)$-mode bipartite system $AB$ can be naturally written in block form according to the splitting between  the subsystems $A$ and $B$:
\begin{equation}
V_{AB}\, =\, \begin{pmatrix} V_A & X \\ X^T & V_B \end{pmatrix}\, .  \label{V explicit}
\end{equation}
According to the same splitting, the matrix $\Omega$ appearing in (\ref{CCR}) takes the form
\begin{equation}
\Omega_{AB}\, =\, \begin{pmatrix} \Omega_A & 0 \\ 0 & \Omega_B \end{pmatrix}\, =\, \Omega_A\oplus\Omega_B\, ,
\end{equation}
with $\Omega_A =  \omega^{\oplus m}$ and $\Omega_B=\omega^{\oplus n}$.

The entanglement properties of a bipartite Gaussian state can thus  be conveniently translated at the level of QCMs. Recall that, in general, a bipartite quantum state $\rho_{AB}$ is separable if and only if it can be written as a convex mixture of product states, $\rho_{AB} = \sum_k p_k ({\sigma_k}_A \otimes {\tau_k}_B)$, with $p_k$ being probabilities~\cite{Werner89}.
For a Gaussian state $\rho^G_{AB}$ of a bipartite continuous variable system, we have then the following.

\begin{lemma}[Proposition 1 in~\cite{Werner01}] \label{sep} $ \\ $
A Gaussian state $\rho^G_{AB}(V_{AB})$ with $(m+n)$-mode QCM $V_{AB}$ is separable if and only if there exist an $m$-mode QCM $\gamma_A \geq i \Omega_A$ and an $n$ mode QCM $\gamma_B \geq i \Omega_B$  such that
\begin{equation} V_{AB} \geq \gamma_A\oplus\gamma_B\, . \label{sep eq} \end{equation}
\end{lemma}

In view of the above result, a QCM $V_{AB}$ satisfying~\eqref{sep eq} for some marginal QCMs $\gamma_{A}$, $\gamma_{B}$ will itself be called {\it separable} from now on. The criterion in (\ref{sep eq}) is necessary and sufficient for separability of QCMs, and can be evaluated numerically via convex optimisation~\cite{Giedke01,eisemi}, however such optimisation runs over both marginal QCMs, hence scaling (polynomially) with both $m$ and $n$.

The first main result of this paper is to show that the necessary and sufficient separability condition~\eqref{sep eq}, for any $m$ and  $n$, can be further simplified. This result is quite neat and of importance in its own right. In particular, it allows us to recast the Gaussian separability problem as a convex optimisation over the marginal QCM of {\it one} subsystem only (say $A$ without loss of generality), resulting in an appreciable reduction of computational resources, especially in case party $A$ comprises a much smaller number of modes than party $B$.

\begin{thm}[Simplified separability condition for an arbitrary QCM] \label{simp sep lemma} $ \\ $
A  QCM $V_{AB}$ of $m+n$ modes is separable if and only if there exists an $m$-mode QCM $\gamma_A \geq i \Omega_A$ such that 
\begin{equation}
V_{AB} \geq \gamma_A\oplus i\Omega_B\, .
\label{simp sep 1}
\end{equation}
In terms of the block form (\ref{V explicit}) of $V_{AB}$, when $V_{B}>i\Omega_{B}$ the above condition is equivalent to the existence of a real matrix $\gamma_A$ satisfying
\begin{equation}
i\Omega_A \leq \gamma_A \leq V_A - X (V_B-i\Omega_B)^{-1} X^T\, .
\label{simp sep 2}
\end{equation}
If $V_{B}-i\Omega_{B}$ is not invertible, we require instead $i\Omega_A \leq \gamma_A \leq V_A - X (V_B+\varepsilon\mathds{1}_{B}-i\Omega_B)^{-1} X^T$ for all $\varepsilon>0$.
\end{thm}

\begin{proof}
Since both sets of QCMs $V_{AB}$ defined by~\eqref{sep eq} and~\eqref{simp sep 1} are clearly topologically closed, we can just show without loss of generality that their interiors coincide. This latter condition can be rephrased as an equivalence between the two following statements: (i) $V_{AB}>\gamma_{A}\oplus \gamma_{B}$ for some QCMs $\gamma_{A},\gamma_{B}$; and (ii) $V_{AB} > \gamma_A\oplus i\Omega_B$ for some QCM $\gamma_{A}$.

Now, once $\gamma_A<V_{A}$ is fixed, the supremum of all the matrices $\gamma_B$ satisfying $V_{AB}>\gamma_{A}\oplus \gamma_{B}$ is given by the Schur complement $(V_{AB}-(\gamma_A\oplus 0_B))/(V_A-\gamma_A)$, as the variational characterisation~\eqref{variational} reveals. Therefore, statement (i) is equivalent to the existence of $i\Omega_{A}\leq\gamma_A<V_{A}$ such that $(V_{AB}-(\gamma_A\oplus 0_B))/(V_A-\gamma_A)> i\Omega_B$. This is the same as to require $V_{AB}> \gamma_A\oplus i\Omega_B$, as the positivity conditions of Lemma~\ref{pos cond} immediately show.

Until now, we have proven that the separability of $V_{AB}$ can be restated as $V_{AB} \geq \gamma_A\oplus i\Omega_B$ for some appropriate QCM $\gamma_{A}$. Employing Lemma~\ref{pos cond}, we see that this is turn equivalent to~\eqref{simp sep 2}, or to its $\varepsilon$-modified version when $V_{B}-i\Omega_{B}$ is not invertible.
\end{proof}

\begin{rem}
It has been recently observed~\cite{Bhat16} that condition~\eqref{simp sep 2} is equivalent to the corresponding Gaussian state $\rho^G_{AB}(V_{AB})$ with QCM $V_{AB}$ being completely extendable with Gaussian extensions. We remind the reader that a bipartite state $\rho_{AB}$ is said to be {\it completely extendable} if for all $k$ there exists a state $\rho_{AB_{1}\cdots B_{k}}$ that is: (i) symmetric under exchange of any two $B_{i}$ systems; and (ii) an extension of $\rho_{AB}$ in the sense that $\text{Tr}_{B_{2}\cdots B_{k}}\rho_{AB_{1}\cdots B_{k}}=\rho_{AB}$. When the original state $\rho_{AB}^G$ is Gaussian, it is natural to consider extensions $\rho^G_{AB_{1}\cdots B_{k}}$ of Gaussian form as well. Interestingly enough, the above Theorem~\ref{simp sep lemma} provides a simple alternative proof of the remarkable fact (also proven in~\cite{Bhat16}) that Gaussian states are separable if and only if completely extendable with Gaussian extensions.
\end{rem}

{
\begin{rem}\label{remulti}
It is worth noticing that both Lemma~\ref{sep} and Theorem~\ref{simp sep lemma} extend straightforwardly to encompass the case of full separability of multipartite Gaussian states. In the case of Lemma~\ref{sep}, this extension was already formulated in~\cite{Werner01,3-mode sep}. As for Theorem~\ref{simp sep lemma}, the corresponding necessary and sufficient condition for the full separability of a $k$-partite QCM $V_{A_{1}\cdots A_{k}}$ would read  $V_{A_{1}\cdots A_{k}}\geq \gamma_{A_{1}}\oplus\ldots\oplus \gamma_{A_{k-1}}\oplus i\Omega_{A_{k}}$ for appropriate QCMs $\gamma_{1},\ldots, \gamma_{k-1}$.
\end{rem}
}

\section{PPT implies separability for 1 vs N-mode Gaussian states -- Revisited}\label{ppt}

We now focus on investigating known and new conditions under which separability becomes equivalent to PPT for Gaussian states, so that the problem of deciding whether a given QCM is separable or not admits a handy formulation.

For any bipartite state $\rho_{AB}$, recall that the PPT criterion provides a useful  necessary condition for separability~\cite{Peres}:
\begin{equation}\label{PeresPPT}\mbox{$\rho_{AB}$ is separable $\ \Rightarrow\ $ $\rho_{AB}^{T_B} \geq 0$}\,,
\end{equation} where the suffix $T_B$ denotes transposition with respect to the degrees of freedom of subsystem $B$ only. In finite-dimensional systems, PPT is also a sufficient condition for separability when $\dim(A) \cdot \dim(B)\leq 6$~\cite{H3}.

In continuous variable systems, the PPT criterion turns out to be also sufficient for separability of QCMs when either $A$ or $B$ is composed of one mode only.

\begin{thm}[PPT is sufficient for Gaussian states of $1$ vs $n$ modes~\cite{Simon00,Werner01}] \label{PPT thm} $ \\ $
Let $V_{AB}$ be a bipartite QCM such that either $A$ or $B$ are composed of one mode only. Then $V_{AB}$ is separable if and only if
\begin{equation}
V_{AB}\, \geq\, \begin{pmatrix} i\Omega_A & 0 \\ 0 & \pm i\Omega_B\end{pmatrix}\, =\, i\Omega_A \oplus (\pm i\Omega_B)\, ,
\label{PPT}
\end{equation}
which amounts to the corresponding Gaussian state being PPT, ${\rho^{G}}^{T_B}_{AB} \geq 0$.
\end{thm}

For completeness, we recall that the partial transpose of an $(m+n)$-mode QCM $V_{AB}$ (i.e., the covariance matrix of the partially transposed density operator ${\rho^{G}}^{T_B}_{AB}$) is given by $V_{AB}^{T_B}= \Theta_B V_{AB} \Theta_B$, where with respect to a mode-wise decomposition on the $B$ subsystem the matrix $\Theta_B$ can be written as $\Theta_B \coloneqq \id_A \oplus \big( \bigoplus_{j=1}^n \zeta\big)_B$, with $\zeta \coloneqq \left(\begin{smallmatrix}1 & 0 \\ 0 & -1 \end{smallmatrix}\right)$~\cite{Simon00}. Accordingly, we can say that the QCM $V_{AB}$ is PPT if and only if $V_{AB}^{T_B}$ is a valid QCM obeying~\eqref{Heisenberg}, which is equivalent to~\eqref{PPT}.

The original proof of Theorem~\ref{PPT thm} came in two steps. Firstly, Simon~\cite{Simon00} proved it in the particular case when both $A$ and $B$ are made of one mode only by performing an explicit analysis of the symplectic invariants of $V_{AB}$; this seminal analysis is quite straightforward to follow and particularly instructive, but eventually a bit cumbersome, since it requires to distinguish between three cases, according to the sign of $\det X$, where $X$ is the off-diagonal block of the QCM $V_{AB}$ partitioned as in (\ref{V explicit}). Later on, Werner and Wolf~\cite{Werner01} reduced the problem for the $1$ vs $n$-mode case with arbitrary $n$ to the $1$ vs $1$-mode case; the proof of this reduction is geometric in nature and rather elegant, but also relatively difficult.

Our purpose in this Section is to use Schur complements to provide the reader with a simple,
direct proof of Theorem~\ref{PPT thm}. Before coming to that, there is a preliminary lemma we want to discuss.

\begin{lemma} \label{2x2 interval}
Let $M,N$ be $2\times 2$ hermitian matrices. There is a real symmetric matrix $R$ satisfying $M\leq R\leq N$ if and only if $M\leq N, N^*$, where $*$ denotes complex conjugation.
\end{lemma}

\begin{proof}
The only complex entry in a $2\times 2$ hermitian matrix is in the off-diagonal element. Suppose without loss of generality that $\Im M_{12}\ge0$ and $\Im N_{12}\le0$ (both conditions in the statement are in fact symmetric under complex conjugation of $M$ or $N$).
It is easy to verify that a $p$ such that $0 \leq p \leq 1$ and $\Im (pM+(1-p)N)_{12}=0$ always exists, and
we see that $R\coloneqq pM+(1-p)N$ is a real symmetric matrix. Moreover, since $R$ belongs to the segment joining $M$ and $N\geq M$ we conclude that $M\leq R\leq N$.
\end{proof}

\begin{rem}
Lemma~\ref{2x2 interval} admits an appealing physical interpretation which also leads to an intuitive proof. This interpretation is based on the fact that $2\times 2$ hermitian matrices can be seen as events in $4$-dimensional Minkowski space-time through the correspondence $x_0 \mathds{1} + \vec{x}\cdot\vec{\sigma} \leftrightarrow (x_0, \vec{x})$. Furthermore, $M\leq N$ translates in Minkowski space-time to `$N$ is in the absolute future of $M$', since the remarkable determinantal identity $\det (x_0 \mathds{1} + \vec{x}\cdot\vec{\sigma}) = x_0^2 - \vec{x}^2$ holds true. Now, the complex conjugation at the matrix level becomes nothing but a spatial reflection with respect to a fixed spatial plane in Minkowski space-time. Thus, our original question is: is it true that whenever both an event $N$ and its spatial reflection $N^*$ are in the absolute future of a reference event $M$ then there is another event $R$ which is: (i) in the absolute future of $M$; (ii) in the absolute past of both $N$ and $N^*$; and (iii) lies right on the reflection plane? The answer is clearly yes, and there is a simple way to obtain it. Start from $M$ and shoot a photon to the location of that event between $N$ and $N^*$ that will happen on the other side of the reflection plane. After some time the photon hits the plane, and this event $R$ clearly satisfies all requirements.
\end{rem}

Now we are ready to give our direct proof of the equivalence between PPT and separability for $1$ vs $n$-mode Gaussian states, leveraging the simplified separability condition of Theorem~\ref{simp sep lemma}.

\begin{proof}[Proof of Theorem~\ref{PPT thm}]
Suppose without loss of generality that $A$ is composed of one mode only. As in the proof of Theorem~\ref{simp sep lemma}, since both sets of QCMs $V_{AB}$ defined by~\eqref{sep eq} and~\eqref{PPT} are topologically closed, we can assume that $V_{AB}$ is in the interior of the PPT set, i.e.~that $V_{AB}>i\Omega_{A}\oplus (\pm i\Omega_{B})$. Our goal will be to show that in this case $V_{AB}$ belongs to the separable set, as characterized by Theorem~\ref{simp sep lemma}. Since $V_{B}-i\Omega_{B}$ is taken to be invertible, the PPT condition reads
\begin{equation*}
V_A - X (V_B \mp i\Omega_B)^{-1} X^T\, \geq\, i\Omega_{A}\, .
\end{equation*}
Now, define $M=i\Omega_{A}$ and $N=V_A - X (V_B + i\Omega_B)^{-1} X^T$, and observe that $N^{*}=V_A - X (V_B - i\Omega_B)^{-1} X^T$. Thanks to Lemma~\ref{2x2 interval}, we can find a real matrix $\gamma_{A}$ such that
\begin{equation*}
V_A - X (V_B \mp i\Omega_B)^{-1} X^T\, \geq\, \gamma_{A}\, \geq\, i\Omega_{A}\, .
\end{equation*}
Choosing the negative sign in the above inequality, we see that the second condition~\eqref{simp sep 2} in Theorem~\ref{simp sep lemma} is met, and therefore $V_{AB}$ is separable.
\end{proof}

\section{Gaussian states that are invariant under partial transpose are separable}\label{inva}

As a further example of application of Theorem~\ref{simp sep lemma}, we study here the separability of a special class of PPT Gaussian states, i.e.~those that are \emph{invariant} under partial transposition of one of the subsystems. This problem has an analogue in finite-dimensional quantum information, already studied in~\cite{sep 2xN}, where it was shown that bipartite states on $\mathds{C}^{2}\otimes \mathds{C}^{d}$ that are invariant under partial transpose on the first system are necessarily separable~\footnote{The proof reported in~\cite{sep 2xN} is rather long, so here we provide a shorter one, again based on Schur complements. A state on $\mathds{C}^{2}\otimes \mathds{C}^{d}$ that is invariant under partial transposition on the first subsystem can be represented in block form as $\rho=\left( \begin{smallmatrix} A & X \\ X & B \end{smallmatrix}\right)$. By a continuity argument, we can suppose without loss of generality that $A>0$. Rewrite $\rho = \left( \begin{smallmatrix} A & X \\ X & XA^{-1}X \end{smallmatrix}\right) + \ket{1}\!\!\bra{1}\otimes (B-XA^{-1}X)$. Both terms are positive by Lemma~\ref{pos cond}. Since the second one is separable, let us deal only with the first one, call it $\tilde{\rho}$. We have $\tilde{\rho} = \id_2 \otimes A^{1/2} \left( \begin{smallmatrix} \id & Y \\ Y & Y^{2} \end{smallmatrix}\right) \id_2 \otimes A^{1/2}$, where $Y\coloneqq A^{-1/2}X A^{-1/2}$ is hermitian. Denoting by $Y = \sum_{i} y_{i} \ket{e_{i}}\!\!\bra{e_{i}}$ its spectral decomposition, we obtain the following manifestly separable representation of $\tilde{\rho}$:
\begin{equation*}
\tilde{\rho} = \id_2 \otimes A^{1/2} \bigg( \sum_{i} \left( \begin{smallmatrix} 1 & y_{i} \\ y_{i} & y_{i}^{2} \end{smallmatrix} \right) \otimes \ket{e_i}\!\!\bra{e_i}\bigg) \id_2 \otimes A^{1/2}\, . \end{equation*}}.
Here we show that for Gaussian states an even stronger statement holds, in that invariance under partial transposition implies separability for any number of local modes.

\begin{cor}
A bipartite Gaussian state $\rho_{AB}^{G}$ that is invariant under partial transposition of one of the two subsystems is necessarily separable.
\end{cor}

\begin{proof}
Without loss of generality, we can assume that the partial transpose on the $B$ system leaves the state invariant. We now show that under the this assumption the separability condition~\eqref{simp sep 2} is immediately satisfied, since the rightmost side is already a real, symmetric matrix. In fact, equating the original QCM~\eqref{V explicit} with the one obtained after partial transpose on the $B$ system, we get the identities $X=X \Theta_B$ and $V_{B}=\Theta_B V_{B}\Theta_B$, where, as previously set, ${\Theta}_B = \bigoplus_{j=1}^{n} \zeta$ according to a mode-wise decomposition of the $B$ system, and $\zeta = \left(\begin{smallmatrix}1 & 0 \\ 0 & -1 \end{smallmatrix}\right)$. As a consequence,
\begin{equation*}
X(V_{B}-i\Omega_{B})^{-1} X^{T} = X \Theta_{B} (V_{B} - i\Omega_{B})^{-1} \Theta_{B} X^{T} = X \left(\Theta_B (V_{B} - i\Omega_{B}) \Theta_B \right)^{-1} X^{T} = X \left(V_{B} + i\Omega_{B} \right)^{-1} X^{T}\, ,
\end{equation*}
where we used also $\Theta_B \Omega_B \Theta_B = -\Omega_B$. This shows that $X(V_{B}-i\Omega_{B})^{-1} X^{T}$ is equal to its complex conjugate, and is therefore (despite appearances) a real symmetric matrix. Hence the separability condition~\eqref{simp sep 2} is satisfied with $\gamma_{A} = V_A - X(V_{B}-i\Omega_{B})^{-1} X^{T}$, which is a legitimate QCM as follows from the bonafide condition~\eqref{Heisenberg} together with Lemma~\ref{pos cond}.
\end{proof}

\section{PPT implies separability for multimode mono-symmetric Gaussian states}\label{symm}

Throughout this Section, we show how the PPT criterion is also necessary and sufficient for deciding the separability of bipartite Gaussian states of $m$ vs $n$ modes that are symmetric under the exchange of any two among the first $m$ modes. These states will be referred to as {\em mono-symmetric} (with respect to the first party $A$).  As can be easily seen, this novel result (see Fig.~\ref{mononucleosi} for a graphical visualisation) is a generalisation of both Theorem~\ref{PPT thm} and of one of the main results in~\cite{Serafini05}, where the subclass of bi-symmetric states was considered instead, bi-symmetric meaning that they are invariant under swapping any two modes either within the first $m$ or within the last $n$ (that is, they are mono-symmetric in both $A$ and $B$).

\begin{thm}[Symplectic localisation of mono-symmetric states]  \label{PPt sym}
Let $\rho^G_{AB}(V_{AB})$ be a mono-symmetric Gaussian state of $m+n$ modes, i.e.~specified by a QCM $V_{AB}$ that is symmetric under the exchange of any two of the $m$ modes of subsystem $A$. Then there exists a local unitary operation on $A$ corresponding to a symplectic transformation $S_A \in \mathrm{Sp}(2m, \mathds{R})$ that transforms $\rho^G_{AB}$ into the tensor product of $m-1$ uncorrelated single-mode Gaussian states $\tilde{\rho}^G_{A_j}(\tilde{V}_{A_j})$ ($j=2,\ldots,m$) and a bipartite Gaussian state $\tilde{\rho}^G_{A_1B}(\tilde{V}_{A_1B})$ of $1$ vs $n$ modes. At the QCM level, this reads
\begin{equation}\label{monolocale}
(S_A \oplus \mathds{1}_B) V_{AB} (S_A^T \oplus \mathds{1}_B) = \bigg(\bigoplus_{j=2}^m \tilde{V}_{A_j}\bigg) \oplus \tilde{V}_{A_1B} \,.
\end{equation}
The separability properties of $V_{AB}$ and $\tilde{V}_{A_1 B}$ are equivalent, in particular $\rho^G_{AB}(V_{AB})$ is separable if and only if it is PPT.
\end{thm}

\begin{proof}
We will prove~\eqref{monolocale} directly at the QCM level, by constructing a suitable local symplectic $S_A$.
By virtue of the symmetry under the exchange of any two modes of subsystem $A$, if we decompose $V_{AB}$ as in~\eqref{V explicit}, the submatrices $V_A$ and $X$ have the following structure:
\begin{equation}\label{strucaz}
V_{A}\, =\, \begin{pmatrix} \alpha & \varepsilon & \ldots & \varepsilon \\[-1ex] \varepsilon & \alpha & & \vdots \\[-1ex] \vdots & & \ddots & \varepsilon \\[-0.5ex] \varepsilon & \ldots & \varepsilon & \alpha \end{pmatrix}\, ,\qquad
X\, =\, \begin{pmatrix} \kappa_{1} & \kappa_{2} & \ldots & \kappa_{n} \\ \kappa_{1} & \kappa_{2} & \ldots & \kappa_{n} \\ \vdots & & & \vdots \\ \kappa_{1} & \kappa_{2} & \ldots & \kappa_{n} \end{pmatrix} \, ,
\end{equation}
where each one of the blocks $\alpha,\varepsilon,\kappa_{j}$ in (\ref{strucaz}) is a $2\times 2$ real matrix, with $\alpha$ and $\varepsilon$ symmetric~\cite{adescaling}.

We can now decompose the real space of the first $m$ modes as $\mathds{R}^{2m}=\mathds{R}^{m} \otimes \mathds{R}^{2}$. According to this decomposition, we may rewrite $V_A$ and $X$ as follows:
\begin{equation}
V_{A}\, =\, \mathds{1}_m\otimes (\alpha-\varepsilon) + m\ket{+}\!\!\bra{+}\otimes \varepsilon\, ,\qquad X\, =\, \sqrt{m}\,\sum_{j=1}^{n} \ket{+}\!\!\bra{j}\otimes \kappa_{j}\, ,
\end{equation}
where $\ket{+}=\frac{1}{\sqrt{m}}\sum_{i=1}^{m}\ket{i}$, with $\{\ket{i}\}_{i=1}^m$ denoting the standard basis for $\mathds{R}^m$. Observe that the symplectic form $\Omega_A$ on subsystem $A$ decomposes accordingly as $\Omega_A = \mathds{1}_m\otimes \omega$. If $O$ is an $m\times m$ orthogonal  matrix such that $O\ket{+}=\ket{1}$, we easily see that on the one hand $O\otimes\mathds{1}_2\ \Omega_A\ O^{T}\otimes\mathds{1}_2 = \Omega_A$, i.e.~$O\otimes\mathds{1}_2$ is symplectic, while on the other hand
\begin{align}
O\otimes\mathds{1}_2\ V_{A}\ O^{T}\otimes\mathds{1}_2\, &=\, \ket{1}\!\!\bra{1} \otimes (\alpha + (m-1)\varepsilon) + \sum_{i=2}^{m} \ket{i}\!\!\bra{i}\otimes (\alpha-\varepsilon)\, =\, \begin{pmatrix} \alpha+(m-1)\varepsilon & 0 & \ldots & 0 \\[-1ex] 0 & \alpha-\varepsilon & & \vdots \\[-1ex] \vdots & & \ddots & 0 \\[-0.5ex] 0 & \ldots & 0 & \alpha-\varepsilon \end{pmatrix}\, , \\[2ex]
O\otimes\mathds{1}_2\ X \, &=\, \sqrt{m}\, \sum_{j=1}^{n} \ket{1}\!\!\bra{j}\otimes \kappa_{j}\, =\, \begin{pmatrix} \sqrt{m}\,\kappa_{1} & \sqrt{m}\, \kappa_{2} & \ldots & \sqrt{m}\,\kappa_{n} \\ 0 & 0 & \ldots & 0 \\ \vdots & & & \vdots \\ 0 & 0 & \ldots & 0 \end{pmatrix}\, .
\end{align}
Therefore, the initial QCM $V_{AB}$ has been decomposed as a direct sum of $m-1$ one-mode QCMs $\tilde{V}_{A_j}=\alpha-\varepsilon$, and of one $(1+m)$-mode QCM $\tilde{V}_{A_1B}$, via a local symplectic operation on subsystem $A$, given precisely by $S_A = O \otimes \mathds{1}_2$. This proves (\ref{monolocale}) constructively. Applying Theorem~\ref{PPT thm}, one then gets immediately that the PPT condition is necessary and sufficient for separability in this case.
\end{proof}

\begin{rem}
This original result yields a substantial enlargement to the domain of validity of PPT as a necessary and sufficient criterion for separability of multimode Gaussian states, reaching beyond any existing literature. In practice, Theorem~\ref{PPt sym} tells us that, in any mono-symmetric Gaussian state, all the correlations (including and beyond entanglement) shared among the whole $m$ modes of $A$ and the whole $n$ modes of $B$ can be  localised onto correlations between a single mode $A_1$ of $A$ vs the whole $B$, by means of a local unitary (symplectic at the QCM level) operation at $A$'s side only. Being unitary, this operation is fully reversible, meaning that the correlations with $B$ can be redistributed back and forth between $A_1$ and the whole set of $A$ modes with no information loss. This also means that quantitative results on any measure of such correlations between $A$ and $B$ encoded in $V_{AB}$ can be conveniently evaluated in the much simpler  $1$ vs $n$-mode normal form $\tilde{V}_{A_1B}$ constructed in the proof Theorem~\ref{PPt sym}, ignoring the $m-1$ uncorrelated modes.

In the special case of $V_{AB}$ being the QCM of a bi-symmetric state, i.e.~with full permutation symmetry within both $A$ and $B$, it is immediate to observe that applying a similar construction by means of a local unitary at $B$'s side as well fully reduces $V_{AB}$ to a two-mode QCM $\tilde{V}_{A_1 B_1}$, with equivalent entanglement properties as the original $V_{AB}$, plus a collection of $m+n-2$ uncorrelated single modes. This reproduces the findings of~\cite{Serafini05}.

Similarly to what discussed in Remark~\ref{remulti}, the results of Theorem~\ref{PPt sym} can also be straightforwardly extended to characterise full separability and, conversely, multipartite entanglement of arbitrary multimode Gaussian states which are partitioned into $k$ subsystems, with the requirement of local permutation invariance within some of these subsystems. It is clear that, by suitable local symplectic transformations, each of those locally symmetric parties can be localised onto a single mode correlated with the remaining parties, thus removing the redundancy in the QCM. Gaussian states of this sort generalise the so-called multi-symmetric states studied in~\cite{moleculo}, where local permutation invariance was enforced within all of the subsystems, resulting in a direct multipartite analogue of bi-symmetric states.

\end{rem}
\section{PPT implies separability for multimode isotropic Gaussian states --  Revisited} \label{mode-wise}

{
It is well known that the PPT criterion is in general sufficient, as well as obviously necessary, for pure bipartite states to be separable~\cite{Peres}. This may be seen by a direct inspection of the Schmidt decomposition of a pure state. Let us note, incidentally, that a stronger statement holds, namely any bound entangled state (in any dimension) must have at least rank $4$~\cite{chen08}.

The Schmidt decomposition theorem is in fact so important that a Gaussian version of it, that is, the determination of a normal form of pure QCMs under local symplectic operations, is of central importance in continuous variable quantum information.
As can be shown at the covariance matrix level~\cite{holwer,giedkemode} or at the density operator level~\cite{botero03}, every pure bipartite Gaussian state $\rho^G_{AB}(V_{AB})$ can be brought into a tensor product of two-mode squeezed vacuum states and single-mode vacuum states by means of local unitaries with respect to the $A$ vs $B$ partition. In particular, by acting correspondingly with local symplectic transformations, any pure QCM $V_{AB}$ (where pure means $\det V_{AB} = 1$) can be transformed into a direct sum of (pure) two-mode squeezed vacuum QCMs and (pure) single-mode vacuum QCMs.
More precisely, at the level of QCMs, one can formulate this fundamental result as follows.

\begin{thm}[Mode-wise decomposition of pure Gaussian states~\cite{holwer,botero03,giedkemode}] \label{mode-wise thm}
Let $V_{AB}$ be a bipartite QCM of $m+n$ modes $A_1,\ldots,A_m,B_1,\ldots,B_n$, assuming $m \leq n$ (with no loss of generality). If $V_{AB}$ is a pure QCM, i.e.~all its symplectic eigenvalues are equal to $1$ (which amounts to $\det V_{AB}=1$), then there exist local symplectic transformations $S_A \in \mathrm{Sp}(2m,\mathds{R})$, $S_B \in \mathrm{Sp}(2n,\mathds{R})$ mapping $V_{AB}$ into the following normal form:
\begin{equation}\label{modewise}
(S_A \oplus S_B) V_{AB} (S_A^T \oplus S_B^T) = \bigoplus_{j=1}^m \bar{V}_{A_jB_j}(r_j)\, \oplus \bigoplus_{k=m+1}^n \mathds{1}_{B_k}\,,
\end{equation}
where $\bar{V}_{A_jB_j} = \begin{pmatrix} c_j \id & s_j {\zeta} \\  s_j {\zeta}  & c_j \id \end{pmatrix}$
with $c_{j}=\cosh(2r_j)$ and $s_j=\sinh(2r_j)$, for a real squeezing parameter $r_j$, is the pure QCM of a two-mode squeezed vacuum state of modes $A_j$ and $B_j$, and $\mathds{1}_{B_k}$ is the pure QCM of the single-mode vacuum state of mode $B_k$.
In particular, with respect to the block form~\eqref{V explicit}, for any pure QCM $V_{AB}$ the marginal QCMs $V_A$ and $V_B$ have matching symplectic spectra, given by
$\vec{\nu}_A=(c_1,\ldots,c_m)$ and $\vec{\nu}_B=(c_1,\ldots,c_m,\underbrace{1,\ldots,1}_{n-m})$.
\end{thm}

Leaving apart its far-reaching applications, in the context of the present paper this result is mainly instrumental for assessing the separability of so-called {\it isotropic} multimode Gaussian states. The QCM of any such state of $m+n$ modes is characterised by the property of having a completely degenerate symplectic spectrum, i.e.~formed of only one distinct symplectic eigenvalue $\nu \geq 1$ (repeated $m+n$ times). This means that the QCM $V_{AB}$ of any isotropic state is proportional by a factor $\nu$ to a pure QCM. Hence, Theorem~\ref{mode-wise thm} tells us that $V_{AB}$ can be brought into a direct sum of two-mode QCMs via a local symplectic congruence (local with respect to any partition into groups of modes $A$ and $B$), as first observed in~\cite{holwer}. Thanks to Theorem~\ref{PPT thm}, this guarantees the following.

\begin{thm} \label{iso thm}
The PPT criterion is necessary and sufficient for separability of all isotropic Gaussian states of an arbitrary number of modes.
\end{thm}

However, notwithstanding the importance of Theorem~\ref{mode-wise} per se, one could strive to seek a more direct way to obtain Theorem~\ref{iso thm}. Our purpose in this Section is in fact to provide an alternative proof of this result, which does not appeal to the mode-wise decomposition theorem at all,  and uses directly Lemma~\ref{sep} instead, leveraging matrix analysis tools such as the notions of matrix means introduced in Section~\ref{secMM}.

Note that, in almost all the remainder of this Section, for a single system of $n$ modes, we will find it more convenient to reorder the vector of canonical operators as $r \coloneqq  (x_1,\ldots,x_n,p_1,\ldots,p_n)$, corresponding to a position-momentum block structure. The symplectic form $\Omega$ appearing in~\eqref{CCR} is accordingly rewritten as
\begin{equation}\label{CCR2}
\Omega \coloneqq
 \begin{pmatrix} 0 & \mathds{1} \\ -\mathds{1} & 0 \end{pmatrix}\,.
  \end{equation}
  We will then write any QCM $V$, as well as any symplectic operation $S$ acting on it, with respect to this alternative block structure, unless explicitly stated otherwise.

Let us start with a preliminary result, equivalent to Proposition 12 of~\cite{manuceau} or Lemma 13 of~\cite{Lami16}. We include a proof for the sake of completeness.

\begin{lemma} \label{QCM geom lemma}
Let $V>0$ be a positive matrix. Then $V$ is a pure QCM if and only if $V=\Omega V^{-1}\Omega^{T}$, and it obeys~\eqref{Heisenberg} if and only if
\begin{equation}\label{WalterWhite}
V\, \geq\, V\# (\Omega V^{-1}\Omega^T)\, .
\end{equation}
\end{lemma}

\begin{proof}
Let the Williamson form of $V$ be given by~\eqref{Williamson}, where (in the convention of this Section) $N=\Lambda \oplus \Lambda$, with $\Lambda\coloneqq(\nu_1,\ldots\nu_n)$. Then we can write
\begin{equation*}
\Omega V^{-1} \Omega^{T}\, =\, \Omega S^{-  T} (\Lambda^{-1} \oplus \Lambda^{-1}) S^{-1}\Omega^{T}\, =\, S \Omega (\Lambda^{-1} \oplus \Lambda^{-1})\Omega^{T} S^{T}\, =\, S (\Lambda^{-1} \oplus \Lambda^{-1}) \Omega \Omega^{T} S^{T}\, =\, S (\Lambda^{-1} \oplus \Lambda^{-1}) S^{T}\, ,
\end{equation*}
where we used in order: (i) the identities $\Omega S^{-  T} = S\Omega$, $S^{-1}\Omega^{T} = \Omega^{T} S^{T}$, all consequences of the defining symplectic identity $S\Omega S^{T}=\Omega$; (ii) the fact that $\Omega$ commutes with $\Lambda^{-1} \oplus \Lambda^{-1}$; and (iii) the orthogonality relation $\Omega\Omega^T=\mathds{1}$. Now the first claim becomes obvious, since $\Lambda=\Lambda^{-1}=\mathds{1}$ if and only if $V$ is a pure QCM. In general, as it can be seen from the above expression, $V$ and $\Omega V^{-1} \Omega^{T}$ are brought in Williamson form by simultaneous congruences with the same symplectic matrix (i.e. $S^{-1}$, in the convention of~\eqref{Williamson}). Hence, the covariance of the geometric mean under congruence ensures that
\begin{equation*}
V\# (\Omega V^{-1}\Omega^{T})\, =\, \left( S \begin{pmatrix} \Lambda & 0 \\ 0 & \Lambda \end{pmatrix} S^{T} \right) \# \left( S \begin{pmatrix} \Lambda^{-1} & 0 \\ 0 & \Lambda^{-1} \end{pmatrix} S^{T} \right)\, =\, S \left( \begin{pmatrix} \Lambda & 0 \\ 0 & \Lambda \end{pmatrix} \# \begin{pmatrix} \Lambda^{-1} & 0 \\ 0 & \Lambda^{-1} \end{pmatrix}\right) S^{T}\, =\, SS^{T}\, ,
\end{equation*}
where the last passage is an easy consequence of the fact that $A\# A^{-1}=\mathds{1}$ for all $A>0$. By comparison with~\eqref{Williamson}, we see that the Heisenberg uncertainty relation $\Lambda\geq \mathds{1}$ can be rephrased as $V=S (\Lambda \oplus \Lambda) S^{T}\geq SS^{T}=V\# (\Omega V^{-1}\Omega^{T})$, which reproduces~\eqref{WalterWhite}, proving the second claim.
\end{proof}

\begin{rem}
From the above proof it is also apparent how, for any positive $A>0$, the matrix $A\# (\Omega A^{-1}\Omega^{T})$ is a pure QCM (independently of the nature of $A$).
\end{rem}

Now we are ready to explain our direct argument to show separability of PPT isotropic Gaussian states, alternative to the use of the mode-wise decomposition.

\begin{proof}[Proof of Theorem~\ref{iso thm}]
We start by rewriting the PPT condition~\eqref{PPT} for a QCM $V_{AB}$ as
\begin{equation*}
\Theta_B V \Theta_B, \geq\, i\Omega\, ,
\end{equation*}
where in the convention of this Section $\Theta_B=\id_A \oplus \left( \begin{smallmatrix} \mathds{1} & 0 \\ 0 & -\mathds{1} \end{smallmatrix}\right)_B$. Thanks to Lemma~\ref{QCM geom lemma}, this becomes in turn
\begin{equation*}
\Theta_B V \Theta_B\, \geq\, (\Theta_B V \Theta_B) \# (\Omega \Theta_B V^{-1} \Theta_B \Omega^{T})
\end{equation*}
and finally
\begin{equation*}
V\, \geq\, V \# (\Theta_B \Omega \Theta_B V^{-1} \Theta_B \Omega^{T}\Theta_B)\, =\, (g V) \# (Z \Omega\, (g V)^{-1}\, \Omega^{T} Z)
\end{equation*}
after conjugating by $\Theta_B$, applying once more the covariance of the geometric mean under congruences, introducing a real parameter $g>0$ (to be fixed later), and defining $Z\coloneqq \mathds{1}_{A}\oplus (-\mathds{1}_{B})=\Theta_B \Omega \Theta_B \Omega^T$. Now, we apply Lemma~\ref{lemma ha=g} to the above expression, obtaining
\begin{equation*}
V\, \geq\, \left(\frac{g V+Z\Omega\, (gV)^{-1}\, \Omega^{T} Z}{2}\right) \# \left( (gV)\, ! \left( Z\Omega\, (gV)^{-1}\, \Omega^{T} Z \right) \right)
\end{equation*}
Although it is not yet transparent, we are done, as the right-hand side of the above inequality is exactly of the form $\gamma_{A}\oplus\gamma_{B}$ when $V$ is the QCM of an isotropic Gaussian state. In fact, let $g>0$ be such that $gV$ is a pure QCM, satisfying $gV=\Omega\, (gV)^{-1}\, \Omega^{T}=\left( \begin{smallmatrix} P & Y \\ Y^{T} & Q \end{smallmatrix} \right)$, where we have now reverted to a block decomposition with respect to the $A$ vs $B$ splitting. Then on the one hand since $Z=\mathds{1}_{A}\oplus (-\mathds{1}_{B})$ we find
\begin{equation*}
\frac{g V+Z\Omega\, (gV)^{-1}\, \Omega^{T} Z}{2}\, =\, \begin{pmatrix} P & 0 \\ 0 & Q \end{pmatrix} ,
\end{equation*}
while on the other hand
\begin{align*}
(gV)\, ! \left( Z\Omega\, (gV)^{-1}\, \Omega^{T} Z \right)\, &=\, 2 \left( (gV)^{-1} + Z\Omega\, (gV)\, \Omega^{T} Z \right)^{-1}\, =\, 2 \Omega \left( \Omega (gV)^{-1} \Omega^{T} + Z (gV) Z \right)^{-1} \Omega^{T}\, =\\
&=\, 2 \Omega \left( gV + Z (gV) Z \right)^{-1} \Omega^{T}\, =\,  \begin{pmatrix} \Omega P^{-1}\Omega^{T} & 0 \\ 0 & \Omega Q^{-1} \Omega^{T}\end{pmatrix} \, ,
\end{align*}
where we used the definition~\eqref{harmonic} of harmonic mean and the fact that $[Z,\Omega]=0$. Putting all together, we find
\begin{equation*}
V\, \geq\, \begin{pmatrix} P & 0 \\ 0 & Q \end{pmatrix}\#\begin{pmatrix} \Omega P^{-1}\Omega^{T} & 0 \\ 0 & \Omega Q^{-1} \Omega^{T}\end{pmatrix}\, =\, \begin{pmatrix} P \# \Omega P^{-1} \Omega^{T} & 0 \\ 0 & Q \# \Omega Q^{-1} \Omega^{T} \end{pmatrix}\, =\, \gamma_{A}\oplus \gamma_{B}\, .
\end{equation*}
Since we already observed that $P \# \Omega P^{-1} \Omega^{T}$ is a QCM for any $P>0$ (and analogously for $Q$), a direct invocation of Lemma~\ref{sep} allows us to conclude the proof.
\end{proof}
}

\section{Entangling Gaussian states via passive optical operations}\label{pass}

Throughout this Section, we finally complete the solution of a problem posed in~\cite{passive} and there addressed under some additional constraints. Let us start by recalling that symplectic operations can be divided into two main categories, namely those such as squeezers that require an exchange of energy between the system and the apparatus, called \emph{active}, and those that can be implemented using only beam splitters and phase plates, called \emph{passive}. A symplectic matrix $K$ represents a passive transformation if and only if it is also \emph{orthogonal}, meaning that $KK^{T}=\id$ (it may be worth adding that symplectic orthogonal transformations form the maximal compact subgroup of the symplectic group).
As it turns out, symplectic orthogonal matrices can be represented in an especially simple form if we resort to a position-momentum block decomposition. Namely, one has the parametrisation~\cite{bucco}
\begin{equation}
K = W^\dag \begin{pmatrix} U & \\ & U^* \end{pmatrix} W\, ,
\end{equation}
where
\begin{equation*}
W \coloneqq \frac{1}{\sqrt{2}} \begin{pmatrix} \id & i\id \\ \id & -i\id \end{pmatrix}
\end{equation*}
and $U$ is a generic, $n\times n$ unitary matrix, with $U^*$ denoting its complex conjugate.

Since the implementation of passive operations is so inexpensive in quantum optics and entangled states so useful for quantum technologies, the question first posed in~\cite{passive} was a natural one: ``What bipartite Gaussian states are such that they can be entangled via a global, passive operation?'' However, in this full generality the problem was left unanswered in~\cite{passive}. Instead, another related question was investigated and answered there, namely whether \emph{distillable} Gaussian entanglement can be produced in the same fashion. For Gaussian states, as mentioned in the Introduction, distillability is well known to be equivalent to non-positivity of the partial transpose~\cite{Giedke01,GiedkeQIC}, so the authors of~\cite{passive} proceeded to identify the class of Gaussian states that can be made to violate the PPT condition with a passive transformation. However, it is important to realise that since PPT and separability are not the same for general multimode Gaussian states, the two questions are a priori different. Here we show that the answer to the original question above turns out to be yet another situation where the PPT condition is necessary and sufficient to ensure separability of Gaussian states. In other words, we will prove that a bipartite Gaussian state that can not be made distillable (i.e.~non-PPT) via passive operations is necessarily separable, and thus it stays separable under the application of said passive operations. Let us start with a technical lemma that we deduce from recent results obtained in~\cite{bhatia15}.

\begin{lemma} \label{lemma eig vs sp eig}
Let $A>0$ be a strictly positive $2n\times 2n$ matrix. Let $\nu_{i} (A)$ and $\lambda_i (A)$ denote its symplectic and ordinary (orthogonal) eigenvalues, respectively, arranged in nondecreasing order. Then
\begin{equation}
\nu_{1}(A)^2\geq \lambda_1(A) \lambda_2(A)\, .
\label{ineq sympl}
\end{equation}
In particular, every positive matrix whose two smallest eigenvalues satisfy $\lambda_1 \lambda_2\geq 1$ is automatically a legitimate QCM.
\end{lemma}

\begin{proof}
From~\cite[Equation (71)]{bhatia15} we deduce $\prod_{j=1}^{k} \nu_{n-j+1}(A)^2\leq \prod_{j=1}^{2k} \lambda_{2n-j+1}(A)$ for all $k=1,\ldots,n$, 
with equality for $k=n$, when both terms equal the determinant of $A$. We can use this observation to deduce that $\prod_{j=1}^k \nu_j(A)^2\geq \prod_{j=1}^{2k} \lambda_j(A)$ for all $k=1,\ldots,n$. The special case $k=1$ yields the claim.
\end{proof}

Now, we are ready to present our strengthening of~\cite[Proposition 1]{passive}.

\begin{thm} \label{abs sep Gauss}
Let $V$ be a bipartite QCM of an $n$-mode system. Then the following are equivalent:
\begin{enumerate}[(i)]
\item $KVK^T$ is separable for all Gaussian passive transformations $K$;
\item $KVK^T$ is PPT for all Gaussian passive transformations $K$; and
\item the two smallest eigenvalues of $V$ satisfy $\lambda_1(V)\lambda_2(V)\geq 1$.
\end{enumerate}
\end{thm}

\begin{proof}
The implication $(i)\Rightarrow (ii)$ is obvious, while $(iii)\Rightarrow (ii)$ already follows from Lemma~\ref{lemma eig vs sp eig} together with the fact that the partial transpose at the level of QCMs is a congruence by orthogonal transformation and thus does not change the ordinary spectrum. One of the main contributions of~\cite{passive} is the proof that $(ii)$ and $(iii)$ are in fact equivalent. In view of this discussion, we have just to show that $(iii)\Rightarrow (i)$. To this end, we will assume that $V$ satisfies $\lambda_1(V)\lambda_2(V)\geq 1$ and construct two local QCMs $\gamma_A, \gamma_B$ that satisfy the hypothesis of the original separability criterion given by Lemma~\ref{sep}. Call $\lambda_1(V)= k$ and observe that if $k\geq 1$ then $V\geq \id=\id_A\oplus \id_B$ and we are done. Otherwise, assume $k<1$ and denote by $\ket{x}$ the normalised eigenvector corresponding to the minimal eigenvalue of $V$, i.e.~$V\ket{x}= k \ket{x}$ and $\braket{x|x}=1$. Since $\lambda_2(V)\geq \frac1k$ and a fortiori $\lambda_i(V)\geq \frac1k$ for all $i\geq 2$, we can write
\begin{equation*}
V\geq k \ket{x}\!\!\bra{x} + \frac1k \left( \id - \ket{x}\!\!\bra{x} \right) .
\end{equation*}
Now, decompose the vector $\ket{x}$ into its $A$ and $B$ components as $\ket{x}=\left(\begin{smallmatrix} \sqrt{p} \ket{y}_A \\ \sqrt{1-p} \ket{z}_B \end{smallmatrix}\right)$, where $0\leq p\leq 1$ and $\braket{y|y}=1=\braket{z|z}$. Then, Lemma~\ref{lemma eig vs sp eig} guarantees that the matrices
\begin{align*}
\gamma_A &\coloneqq k \ket{y}\!\!\bra{y} + \frac1k \left( \id - \ket{y}\!\!\bra{y} \right) \\
\gamma_B &\coloneqq k \ket{z}\!\!\bra{z} + \frac1k \left( \id - \ket{z}\!\!\bra{z} \right)
\end{align*}
are legitimate QCMs. Then showing $V_{AB} - \gamma_A\oplus \gamma_B\geq 0$ would complete our proof. By direct computation, we find
\begin{align*}
V_{AB} - \gamma_A\oplus \gamma_B &\geq -\left( \frac1k - k\right) \begin{pmatrix} p \ket{y}\!\!\bra{y} & \sqrt{p(1-p)} \ket{y}\!\!\bra{z} \\ \sqrt{p(1-p)} \ket{z}\!\!\bra{y} & (1-p) \ket{z}\!\!\bra{z} \end{pmatrix} + \frac1k \begin{pmatrix} \id & 0 \\ 0 & \id \end{pmatrix} \\
&\quad - \begin{pmatrix}  -\left( \frac1k - k\right) \ket{y}\!\!\bra{y} + \frac1k \id & 0 \\ 0 & -\left( \frac1k - k\right) \ket{z}\!\!\bra{z} + \frac1k \id \end{pmatrix} \\
&= \left( \frac1k - k\right) \begin{pmatrix} (1-p) \ket{y}\!\!\bra{y} & - \sqrt{p(1-p)} \ket{y}\!\!\bra{z} \\ -\sqrt{p(1-p)} \ket{z}\!\!\bra{y} & p \ket{z}\!\!\bra{z} \end{pmatrix} \\
&= \left( \frac1k - k\right) \begin{pmatrix} \sqrt{1-p} \ket{y} & -\sqrt{p} \ket{z} \end{pmatrix}^T \begin{pmatrix} \sqrt{1-p} \ket{y} & -\sqrt{p} \ket{z} \end{pmatrix} \\
&\geq 0\, .
\end{align*}
\end{proof}

\begin{rem}
In some sense, one can think of the question posed in~\cite{passive} and answered here in Theorem~\ref{abs sep Gauss} as a continuous variable analogue of the \emph{absolute separability} problem in finite-dimensional quantum information, which asks for the characterisation of those spectra $\sigma=(\lambda_1,\ldots,\lambda_{d d'})$ such that every bipartite quantum state on $\mathds{C}^d\otimes \mathds{C}^{d'}$ with spectrum $\sigma$ is separable~\cite{kus01}. For a recent review of the state of the art, we refer the reader to~\cite{abs sep review}.
A suggestive argument concerning this analogy goes as follows. An arbitrary unitary transformation $\rho\mapsto U\rho U^\dag$ corresponds to an internal time evolution according to some unknown Hamiltonian. Then, the absolutely separable states are exactly those bipartite states whose correlations are so weak that they can not be made entangled by any internal evolution. In the case of continuous variable quantum systems, one may hold the free-field Hamiltonian $\mathcal{H}=\frac12 r^Tr$ as the privileged one, so that it makes sense to restrict oneself to those unitary evolutions that preserve this particular Hamiltonian. If the original state is Gaussian and the unitaries are generated by quadratic Hamiltonians, so that they are represented by symplectic matrices, preserving the free-field Hamiltonian is the defining feature of passive transformations, and one obtains exactly the problem we solved here.

As is often the case, the technical details and the nature of the solution are  simpler in the Gaussian realm. We found that the condition for being `absolutely separable' in the Gaussian sense is expressed by a simple inequality involving only the two smallest ordinary eigenvalues of the QCM, and that there are no `absolutely PPT' states that are not `absolutely separable' too. This latter equivalence has been conjectured to hold for the original problem in discrete-variable systems as well, but so far only partial answers are available. Namely, the conditions for absolute PPT-ness can be written explicitly~\cite{hildebrand07}, but whether or not they imply absolute separability is in general unknown. However, the answer to this latter question has been shown to be affirmative for the case of two qubits~\cite{verstraete01} and more recently for qubit-qudit systems~\cite{johnston13}.
\end{rem}
\section{Summary and outlook}\label{outro}

In this work we advanced the mathematical and physical study of separability and entanglement distillability in Gaussian states of continuous variable quantum systems. Based on the properties of Schur complements and other matrix analysis tools, we obtained a simplified necessary and sufficient condition for the separability of all multimode Gaussian states, requiring optimisation over the set of local covariance matrices of one subsystem only. Exploiting this result, we presented a compact proof
of the equivalence between PPT and separability for  $1$ vs $n$-mode Gaussian states, a seminal result in continuous variable quantum information theory~\cite{Simon00,Werner01}, as well as extended the criterion to multimode classes of so-called mono-symmetric and isotropic Gaussian states, through novel derivations.
Furthermore, we completed  the investigation of entanglement generation under passive operations by
extending seminal results~\cite{passive} to consider the generation of any, possibly PPT, Gaussian entangled state: in this context we showed that, if passive operations can not turn an initial Gaussian state into a non-PPT one, then no PPT entanglement can be generated through them either. This can be interpreted as establishing the equivalence between absolute separability and absolute PPT-ness in the Gaussian world.
Side results of our analysis include a novel proof that
Gaussian states invariant under partial transposition are separable,
as well as an independent proof of the equivalence between Gaussian separability and complete extendability with Gaussian extensions~\cite{Bhat16}.

In the context of this paper, and with the methods illustrated in this study,
it would be interesting to research more general combinations of symmetries and conditions on the symplectic
spectra of quantum covariance matrices whereby the sufficiency of the PPT separability criterion might be further extended. For instance, is it possible to obtain a Gaussian analogue of the results in~\cite{chen08}, whereby bound entangled Gaussian states can only exist given some simple condition on their symplectic rank? In our studies, both for mono-symmetric and isotropic states, large degeneracies in their symplectic spectra (for the marginal covariance matrix of one subsystem, and for the global covariance matrix of the bipartite system, respectively) played a key role in proving the sufficiency of PPT for separability. It would be desirable to provide a full systematic characterisation of such requirements, possibly drawing inspiration from and/or shedding new insight on the Gaussian quantum marginal problem~\cite{tyc}.

Finally, let us stress how matrix analysis tools such as those heavily hammered in this paper have already been proven useful for qualitative and quantitative analysis of entanglement and other correlations, including Einstein-Podolsky-Rosen steering, in general states of continuous variable systems~\cite{Giedke01,giedkemode,nogo1,nogo2,nogo3,eisemi,wise,Adesso12,gian,Simon16,Lami16,anders}. Aside from the fact that very powerful analytical results can be proven with relative simplicity using these tools, it is important to remark once more that the characterisations we provided of the separability problem, as well as the variational characterisation of the Schur complement and related problems, can be straightforwardly recast as {\em semidefinite programs}~\cite{eisemi}, thus leading to efficient numerical methods to witness inseparability and entanglement distillability in general multimode Gaussian or non-Gaussian states based on covariance matrices. 
We will explore these and other applications in further studies.

\begin{acknowledgments}
GA warmly acknowledges highly stimulating interactions with organisers, lecturers, and participants at the 2nd IMSc School on Quantum Information (Chennai, India, December 2016), during which this work was completed, and in particular very fruitful discussions with R.~Simon, M.~Banik, R.~Sengupta, and A.~Nayak on topics related to this paper.
We acknowledge financial support from the European Union under the European Research Council (StG GQCOP No.~637352 and AdG IRQUAT No.~267386) and the European Commission (STREP RAQUEL No.~FP7-ICT-2013-C-323970), the Foundational Questions Institute (fqxi.org) Physics of the Observer Programme (Grant No.~FQXi-RFP-1601), the Spanish MINECO (Project no. FIS2013-40627-P and no. FIS2016-86681-P), and the Generalitat de Catalunya (CIRIT Project No. 2014 SGR 966). AS acknowledges financial support from EPSRC through grant EP/K026267/1.
\end{acknowledgments}

\providecommand \doibase [0]{http://dx.doi.org/}%
\providecommand \dois[2]{\href{\doibase#1}{#2}}


\begin{thebibliography}{99}

\bibitem{introeisert}
J. Eisert and M. B. Plenio, Introduction to the basics of entanglement theory in continuous variable systems,
\dois{10.1142/S0219749903000371}{Int. J. Quant. Inf. {\bf 1}, 479 (2003)}.


\bibitem{biblioparis} A. Ferraro, S. Olivares, and M. G. A. Paris, \href{https://arxiv.org/abs/quant-ph/0503237}{{\it Gaussian states
in continuous variable quantum information} (Bibliopolis, Napoli, 2005)}

\bibitem{adesso07}
G. Adesso and F. Illuminati, Entanglement in continuous-variable systems: recent advances and current perspectives, \dois{10.1088/1751-8113/40/28/S01}{J. Phys. A: Math. Theor. {\bf 40} 7821 (2007)}.

\bibitem{weedbrook12} C. Weedbrook, S. Pirandola, R. Garc\'ia-Patr\'on, N. J. Cerf, T. C. Ralph, J. H. Shapiro, and S. Lloyd,
Gaussian quantum information, \dois{10.1103/RevModPhys.84.621}{Rev. Mod. Phys. {\bf 84}, 621 (2012)}.

\bibitem{adesso14} G. Adesso, S. Ragy, and A. R. Lee, Continuous variable quantum information:
Gaussian states and beyond, \dois{10.1142/S1230161214400010}{Open Syst. Inf. Dyn. {\bf 21}, 1440001 (2014)}.

\bibitem{bucco}A. Serafini, \href{https://www.crcpress.com/Quantum-Continuous-Variables-A-Primer-of-Theoretical-Methods/Serafini/p/book/9781482246346}{\em Quantum Continuous Variables: A Primer of Theoretical Methods} (CRC Press, 2017).

\bibitem{Werner01} R. F. Werner and M. M. Wolf, Bound entangled Gaussian states, \dois{10.1103/PhysRevLett.86.3658}{Phys. Rev. Lett. {\bf 86}, 3658 (2001)}.

\bibitem{Giedke01} G. Giedke, B. Kraus, M. Lewenstein, and J. I. Cirac, Entanglement criteria for all bipartite Gaussian states,
\dois{10.1103/PhysRevLett.87.167904}{Phys. Rev. Lett. {\bf 87}, 167904 (2001)}.

\bibitem{HorodeckiBound}
M. Horodecki, P. Horodecki, and R. Horodecki, Mixed-state entanglement and distillation: Is there a ``bound'' entanglement in nature?,
 \dois{10.1103/PhysRevLett.80.5239}{Phys. Rev. Lett. {\bf 80}, 5239 (1998)}.

\bibitem{nogo1} J. Eisert, S. Scheel, and M. B. Plenio, Distilling Gaussian states with Gaussian operations is impossible, \dois{10.1103/PhysRevLett.89.137903}{Phys. Rev. Lett. {\bf 89}, 137903 (2002)}.

\bibitem{nogo2} J. Fiur\'a\u{s}ek, Gaussian transformations and distillation of entangled Gaussian states, \dois{10.1103/PhysRevLett.89.137904}{Phys. Rev. Lett. {\bf 89}, 137904 (2002)}.

\bibitem{nogo3} G. Giedke and J. I. Cirac, Characterization of Gaussian operations and distillation of Gaussian states, \dois{10.1103/PhysRevA.66.032316}{Phys. Rev. A {\bf 66}, 032316 (2002)}.

\bibitem{GiedkeQIC} G. Giedke, L.-M. Duan, J. I. Cirac, and P. Zoller, Distillability Criterion for all Bipartite Gaussian States,
\href{https://arxiv.org/abs/quant-ph/0104072}{Quant. Inf. Comput. {\bf 1}, 79 (2001)}.

\bibitem{Peres} A. Peres, Separability criterion for density matrices, \dois{10.1103/PhysRevLett.77.1413}{Phys. Rev. Lett. {\bf 77}, 1413 (1996)}.

\bibitem{Simon00} R. Simon, Peres-Horodecki separability criterion for continuous variable systems, \dois{10.1103/PhysRevLett.84.2726}{Phys. Rev. Lett. {\bf 84}, 2726 (2000)}.

\bibitem{Serafini05} A. Serafini, G. Adesso, and F. Illuminati,  Unitarily localizable entanglement of Gaussian states, \dois{10.1103/PhysRevA.71.032349}{Phys. Rev. A {\bf 71}, 032349 (2005)}.

\bibitem{holwer}A. S. Holevo and R. F. Werner, Evaluating capacities of bosonic Gaussian channels, \dois{10.1103/PhysRevA.63.032312}{Phys. Rev. A {\bf 63}, 032312 (2001)}.

\bibitem{giedkemode} G. Giedke, J. Eisert, J. I. Cirac, and M. B. Plenio,
Entanglement transformations of pure Gaussian states, \href{https://arxiv.org/abs/quant-ph/0301038}{Quant. Inf. Comput. {\bf 3}, 211 (2003)}.

\bibitem{botero03} A. Botero and B. Reznik, Modewise entanglement of Gaussian states, \dois{10.1103/PhysRevA.67.052311}{Phys. Rev. A {\bf 67}, 052311 (2003)}.

\bibitem{eisemi} P. Hyllus and J. Eisert, Optimal entanglement witnesses for continuous-variable systems, \dois{10.1088/1367-2630/8/4/051}{New J. Phys. 8, 51 (2006)}.

\bibitem{gian} I. Kogias, A. R. Lee, S. Ragy, and G. Adesso, Quantification of Gaussian quantum steering, \dois{10.1103/PhysRevLett.114.060403}{Phys. Rev. Lett. 114, 060403 (2015)}.

\bibitem{Simon16} G. Adesso and R. Simon, Strong subadditivity for log-determinant of covariance
matrices and its applications, \dois{10.1088/1751-8113/49/34/34LT02}{J. Phys. A: Math. Theor. {\bf 49}, 34LT02
(2016)}.

\bibitem{Lami16} L. Lami, C. Hirche, G. Adesso, and A. Winter, Schur complement inequalities for covariance matrices and monogamy of quantum correlations, \dois{10.1103/PhysRevLett.117.220502}{Phys. Rev. Lett. {\bf 117}, 220502 (2016)}.

\bibitem{sep 2xN} B. Kraus, J. I. Cirac, S. Karnas, and M. Lewenstein, Separability in $2\times N$ composite quantum systems, \dois{10.1103/PhysRevA.61.062302}{Phys. Rev. A {\bf 61}, 062302 (2000)}.

\bibitem{passive} M. M. Wolf, J. Eisert, and M. B. Plenio, Entangling power of passive optical elements, \dois{10.1103/PhysRevLett.90.047904}{Phys. Rev. Lett. {\bf 90}, 047904 (2003)}.

\bibitem{ZHANG05} F. Zhang, \href{http://www.springer.com/gp/book/9780387242712}{\emph{The Schur complement and its applications} (Springer, 2005)}.

\bibitem{BHATIA} R. Bhatia, \href{http://press.princeton.edu/titles/8445.html}{{\it Positive Definite Matrices}, Princeton Series in Applied Mathematics (Princeton University Press, 2009)}.

\bibitem{parallel sum} W. N. Anderson Jr. and R. J. Duffin, Series and parallel additions of matrices, \dois{10.1016/0022-247X(69)90200-5}{J. Math. Anal. Appl. {\bf 26}, 576-594 (1969)}.

\bibitem{ando79} T. Ando, Concavity of certain maps on positive definite matrices and applications to Hadamard products, \dois{10.1016/0024-3795(79)90179-4}{Linear Algebra Appl. {\bf 26}, 2013 (1979)}.

\bibitem{geometric mean} W. Pusz and S. L. Woronowicz, Functional calculus for sesquilinear forms and the purification map, \dois{10.1016/0034-4877(75)90061-0}{Rep. Math. Phys. {\bf 8}, 159-170 (1975)}.

\bibitem{simon94} R. Simon, N. Mukunda, and B. Dutta, Quantum-noise matrix for multimode systems: U(n) invariance, squeezing, and normal forms, \dois{10.1103/PhysRevA.49.1567}{Phys. Rev. A {\bf 49}, 1567
(1994)}.

\bibitem{pramana} Arvind, B. Dutta, N. Mukunda, and R. Simon, The real symplectic groups in quantum mechanics and optics,
\dois{10.1007/BF02848172}{Pramana  {\bf 45}, 471 (1995)}.

\bibitem{willy} J. Williamson, On the algebraic problem concerning the normal forms of linear dynamical systems, \dois{10.2307/2371062}{Am. J. Math. {\bf 58}, 141 (1936)}.

\bibitem{willysim} R. Simon, S. Chaturvedi, and V. Srinivasan, Congruences and canonical forms for a positive matrix: application to the Schweinler-Wigner extremum principle,
\dois{10.1063/1.532913}{J. Math. Phys. {\bf 40}, 3632 (1999)}.

\bibitem{Werner89} R. F. Werner, Quantum states with Einstein-Podolsky-Rosen correlations admitting a hidden-variable model, \dois{10.1103/PhysRevA.40.4277}{Phys. Rev. A  {\bf 40}, 4277  (1989)}.

\bibitem{Bhat16} B. V. Rajarama Bhat, K. R. Parthasarathy, and R. Sengupta, On the equivalence of separability and extendability of quantum states, \dois{10.1142/S0129055X1750012X}{Rev. Math. Phys. {\bf 29}, 1750012 (2017)}.

\bibitem{3-mode sep} G. Giedke, B. Kraus, M. Lewenstein, and J. I. Cirac, Separability properties of three-mode Gaussian states, \dois{10.1103/PhysRevA.64.052303}{Phys. Rev. A {\bf 64}, 052303 (2001)}.

\bibitem{H3} M. Horodecki, P. Horodecki, and R. Horodecki, Separability of mixed states: necessary and sufficient conditions, \dois{10.1016/S0375-9601(96)00706-2}{Phys. Lett. A {\bf 223}, 1 (1996)}.

\bibitem{adescaling} G. Adesso, A. Serafini, and F. Illuminati, Quantification and scaling of multipartite entanglement in continuous variable systems,
\dois{10.1103/PhysRevLett.93.220504}{Phys. Rev. Lett. {\bf 93}, 220504 (2004)}.

\bibitem{moleculo} G. Adesso and F. Illuminati, Genuine multipartite entanglement of symmetric Gaussian states: Strong monogamy, unitary localization, scaling behavior, and molecular sharing structure, \dois{10.1103/PhysRevA.78.042310}{Phys. Rev. A {\bf 78}, 042310 (2008)}.

\bibitem{chen08} L. Chen and Y.-X. Chen, Rank-three bipartite entangled states are distillable, \dois{10.1103/PhysRevA.78.022318}{Phys. Rev. A {\bf 78}, 022318 (2008)}.

\bibitem[42]{42}{D. Adams, \href{http://openlibrary.org/books/OL10488032M/The_Hitch_Hiker's_Guide_to_the_Galaxy}{{\it The hitchhiker's guide to the galaxy} (Harmony Books, New York, 1979)}.}

\bibitem{manuceau}J. Manuceau and A. Verbeure, Quasi-free states of the C.C.R.---Algebra and Bogoliubov transformations, \dois{10.1007/BF01654283}{Commun. Math. Phys. {\bf 9}, 293 (1968)}.

\bibitem{bhatia15} R. Bhatia and T. Jain, On symplectic eigenvalues of positive definite matrices, \dois{10.1063/1.4935852}{J. Math. Phys. {\bf 56}, 112201 (2015)}.


\bibitem{kus01} M. Ku{\'s}, and K. {\.Z}yczkowski, Geometry of entangled states, \dois{10.1103/PhysRevA.63.032307}{Phys. Rev. A {\bf 63}, 032307 (2001)}.

\bibitem{abs sep review} S. Arunachalam, N. Johnston, and V. Russo, Is absolute separability determined by the partial transpose?, \href{https://arxiv.org/abs/1405.5853}{Quant. Inf. Comput. {\bf 15}, 0694-0720 (2015)}.

\bibitem{hildebrand07} R. Hildebrand, Positive partial transpose from spectra, \dois{10.1103/PhysRevA.76.052325}{Phys. Rev. A {\bf 76}, 052325 (2007)}.

\bibitem{verstraete01} F. Verstraete, K. Audenaert, and B. {De Moor}, Maximally entangled mixed states of two qubits, \dois{10.1103/PhysRevA.64.012316}{Phys. Rev. A {\bf 64}, 012316 (2001)}.

\bibitem{johnston13} N. Johnston, Separability from spectrum for qubit-qudit states, \dois{10.1103/PhysRevA.88.062330}{Phys. Rev. A {\bf 88}, 062330 (2013)}.


\bibitem{tyc} J. Eisert, T. Tyc, T. Rudolph, and B. C. Sanders, Gaussian quantum marginal problem, \dois{10.1007/s00220-008-0442-4}{Commun. Math. Phys. {\bf 280}, 263 (2008)}.

\bibitem{anders} J. Anders and A. Winter, Entanglement and separability of quantum harmonic oscillator systems at finite temperature, \href{https://arxiv.org/abs/0705.3026}{Quant. Inf. Comput. {\bf 8}, 0245-0262 (2008)}.  

\bibitem{wise} H. M. Wiseman, S. J. Jones, and A. C. Doherty, Steering, entanglement, nonlocality, and the Einstein-Podolsky-Rosen paradox, \dois{10.1103/PhysRevLett.98.140402}{Phys. Rev. Lett. 98, 140402 (2007)}.

\bibitem{Adesso12} G. Adesso, D. Girolami, and A. Serafini, Measuring Gaussian quantum information and correlations using the R\'enyi entropy of order 2, \dois{10.1103/PhysRevLett.109.190502}{Phys. Rev. Lett. {\bf 109}, 190502 (2002)}.





\end{thebibliography}
\end{document}